\documentclass[final,1p,times]{elsarticle}
\usepackage{amssymb}
\usepackage{amsthm}
\usepackage[bold]{complexity}
\usepackage{graphicx}
\usepackage{amsmath}

\newtheorem{thm}{Theorem}
\newtheorem{lem}[thm]{Lemma}
\newtheorem{cor}[thm]{Corollary}

\journal{Discrete Applied Mathematics}

\begin{document}

\begin{frontmatter}

\title{Bounding the Feedback Vertex Number of Digraphs in Terms of Vertex Degrees}

\author{Hermann Gruber}

\address{knowledgepark AG \\ 
Leonrodstr. 68, M\"unchen, Germany}
\ead{info@hermann-gruber.com}
\begin{abstract}
The Tur{\'a}n bound~\cite{Turan41} is a famous result in graph theory, 
which relates the independence number 
of an undirected graph to its edge density.  
Also the Caro--Wei inequality~\cite{Caro79,Wei81}, 
which gives a more refined bound in terms of the vertex degree 
sequence of a graph, might be regarded today as a classical result. 
We show how these statements can be generalized to directed   
graphs, thus yielding a bound on directed 
feedback vertex number in terms of vertex outdegrees and in 
terms of average outdegree, respectively.
\end{abstract}


\end{frontmatter}

\section{Introduction}
Not only in discrete mathematics, generalizing existing concepts and proofs 
has always been a guiding theme for research. The great mathematician 
Henri Poincar\'{e} even considered this as the {\it leitmotiv} 
of all mathematics.\footnote{Poincar\'{e}'s original phrasing was more poetic: ``La math\'{e}matique est l'art de donner le m\^{e}me nom \`{a} des choses diff\'{e}rentes.''~\cite[p.~29]{Poincare08}.}
 
In particular, many results from graph theory were generalized 
to weighted graphs, digraphs, or 
hypergraphs. Sometimes, providing such generalizations is an easy exercise; 
in other cases, the main difficulty lies in formulating the 
``right generalization'' of the original theorem.  
An additional obstacle is imposed if the result we intend to generalize
allows for several proofs or equivalent reformulations. Then 
there are many roads to potential generalizations to explore, 
and selecting the most promising one can be difficult. However,
once the proper generalizations of the used notions are found, 
the more general proof often runs very much along the same lines.      

As we shall see, one such example is the Tur\'{a}n bound~\cite{Turan41}, 
which gives the number of edges that a graph of order~$n$ can have when 
forbidding $k$-cliques as subgraphs. It allows for many different 
proofs and equivalent reformulations, see~\cite{Aigner95}. 
A dual version of Tur\'{a}n's bound, 
regarding the size of independent sets, 
was refined by Caro~\cite{Caro79} and Wei~\cite{Wei81}. 
Their result has subsequently been generalized, 
by replacing the independent sets with less 
restricted induced subgraphs~\cite{AKS87},
respectively by replacing the concept of a graph 
with more general notions, namely weighted graphs~\cite{STY03} 
and hypergraphs~\cite{CT91}. 
Here we complement these efforts by providing 
a generalization of the Caro--Wei bound to the case of digraphs.
From an algorithmic perspective, the new result
gauges a simple greedy heuristic for the minimum 
directed feedback vertex set problem. In this way, the main result 
of this paper yields a 
formalized counterpart to the intuition that the minimum (directed) 
feedback vertex number of sparse digraphs cannot be ``overly large''. 

\section{Preliminaries}
We assume the reader is familiar with basic notions in the 
theory of digraphs, as contained in textbooks 
such as~\cite{BG00}. 
Nevertheless, we briefly recall the most important 
notions in the following. 
A {\em digraph}~$D=(V,A)$ consists of a finite set, referred to 
as the set~$V(G) = V$ of {\em vertices}, and of an irreflexive binary 
relation on~$V(G)$, referred to as the set of 
{\em arcs}~$A(G) = A \subset V \times V$. The cardinality of the 
vertex set is referred to as the {\em order} of~$D$. 
In the special case where the arc relation of a digraph 
is symmetric, we also speak of an {\em (undirected) graph}.
For a vertex~$v$ in a digraph~$D$, define its 
{\em out-neighborhood} as $N^+(v) = \{u \in V \mid (u,v) \in A, u \neq v \}$,
and its {\em out-degree} as $d^+(v)=|N^+(v)|$. 
In-neighborhood and in-degree are defined analogously, 
and denoted by $N^-(v)$ 
and~$d^{-}(v)$, respectively. The {\em degree} $d(v)$ of~$v$
is then defined as $d(v) = |N^{-}(v) \cup N^{+}(v)|$, and 
the {\em total degree} of~$v$ is defined as $|N^{-}(v)| + |N^{+}(v)|$. 
We note that our definition of vertex degree agrees (on undirected 
graphs) with the standard usage of this notion in the theory of 
undirected graphs, see e.g.~\cite{Diestel06}.     
For a subset of vertices $U \subseteq V$ of the 
digraph $D=(V,A)$, the {\em subdigraph induced by~$U$} is the 
digraph~$(U,A|_{U \times U})$ obtained by reducing the 
vertex set to~$U$ and by restricting the arc set to 
the relation induced by~$A$ on~$U$. 
If a digraph~$H$ can be 
obtained in this way by appropriate restriction of the vertex set 
of the digraph~$D$, we say~$H$ is an {\em induced subdigraph} of~$D$.
A {\em simple path} in a digraph is a sequence of $k\ge 1$ arcs
$(v_1,w_1)(v_2,w_2)\cdots(v_k,w_k)$, such that 
for all $1 \le i < k$ holds $w_i = v_{i+1}$ and 
all start-vertices~$v_i$ are distinct. 
If furthermore~$w_k=v_1$, we speak of a {\em cycle}. 
In particular, notice that  
each pair of opposite arcs $(v,w)(w,v)$ in a 
digraph amounts to a cycle. This convention is commonly 
used in the theory of digraphs, compare~\cite{BG00}.
 
A digraph containing no cycles is called {\em acyclic}, 
or a {\em directed acyclic graph (DAG)}.  
For a vertex subset~$U$ of a digraph~$D$, 
if the subdigraph induced by~$U$ is acyclic, 
then we call~$U$ an {\em acyclic set}. 
In particular, if $D[U]$ contains no arcs at all, then~$U$ 
is called an {\em independent set}. The maximum cardinality 
among all independent sets in~$D$ 
is called the {\em independence number} of~$D$. 
Tur{\'a}n proved the following bound on the independence 
number of undirected graphs:
\begin{thm}
Let $D=(V,A)$ be an undirected graph of order~$n$ and of average 
degree~$\overline{d}$. Then~$D$ contains 
an independent set of size at least 
$\left(\overline{d} + 1\right)^{-1} \cdot n$. 
\end{thm}
Caro~\cite{Caro79}, and, independently, Wei~\cite{Wei81} 
proved the following refined bound: 
\begin{thm}
Let $D=(V,A)$ be an undirected graph of order~$n$. 
Then~$D$ contains an independent set  
of size at least $\sum_{v\in V} \left(d(v) +1\right)^{-1}$.
\end{thm}

A set~$F$ of vertices in a digraph~$D=(V,A)$ is called a 
{\em feedback vertex set} if $V \setminus F$ is an 
acyclic set. The {\em feedback vertex number} $\tau_0(D)$ 
of~$D$ is defined as the minimum cardinality among all 
feedback vertex sets for~$D$. A simple observation is 
that for a digraph~$D$ of order $n$, the cardinality of a 
maximum acyclic set equals $n-\tau_0$. 

\section{Directed Feedback Vertex Sets and Vertex Degrees}

Quite recently, several new algorithms were 
devised for exactly solving the minimum directed feedback vertex 
set problem~\cite{CLLOR08,Razgon07}.
But all known exact algorithms for this problem 
share the undesirable feature that
their worst-case running time is exponential---in the 
order~$n$ of the input graph, or at least in the size 
of the feedback vertex number~$\tau_0$. 
This is not surprising as the problem has been known for a long time 
to be $\NP$-complete, see~\cite{GJ79}. 

Here, we consider the following simple 
greedy heuristic for finding a large acyclic set, 
and hence a small feedback vertex set,  
in a digraph~$D$. We call the algorithm 
{\em Min-Greedy}, in accordance with a
homonymous greedy heuristic on undirected
graphs for finding a large independent set, 
compare~\cite{Griggs83,HR97}.

Starting with $D_1=D$, we inductively define a 
sequence of digraphs $D_i$, $i\ge 1$, by first 
choosing a vertex $v_i$, such that $v_i$ 
has minimum outdegree in $D_i$, and then 
deleting $v_i$, along with its out-neighborhood 
in~$D_i$, to obtain the digraph $D_{i+1}$. 
We proceed in doing so until the  
vertex set of $D_i$ is empty, and remember the 
vertices $v_i$ selected in each turn. These 
vertices form the set $S=\{v_1,v_2,\ldots v_r\}$, 
which is the result finally returned by 
the procedure.

Before we analyze the quality of the above heuristic, 
we shall first prove its soundness.
\begin{lem}\label{lem:soundness}
Let $D$ be a digraph. Then the set~$S$ 
returned by {\em Min-Greedy} on input $D$ 
is an acyclic set in~$D$.
\end{lem}
\begin{proof}
Using the notions from the description of the algorithm,
it suffices to show that for all $v_j,v_k \in S$, the 
condition $j<k$ implies that the digraph~$D$ has 
no arc $(v_j,v_k)$. This claim implies 
that along every simple path in $D[S] = D[\{v_1,v_2,\ldots,v_r\}]$, 
the vertex indices must appear in 
decreasing order, thus ruling out the possibility 
of a cycle in~$D[S]$. 

To prove the claim, observe first that for all $k>j$, 
$D_k$ is an induced subdigraph of $D_{j+1}$. Thus starting 
with $D_{j+1} = D_j - (\{v_j\} \cup N^+(v_j)) $, 
no vertex in the out-neighborhood of $v_j$ is present 
in any of the subsequent digraphs. But $v_k$ is selected 
from~$G_k$, hence is present in $G_k$ and cannot be in the 
out-neighborhood of~$v_j$.
\end{proof}
Observe that the proof of Lemma~\ref{lem:soundness} does 
not depend on the choice of a vertex of minimum 
degree in~$D_i$ for~$v_i$---the algorithm is sound 
if we choose any vertex in~$D_i$ for~$v_i$. 
Now we are ready to state our main result.

\begin{thm}\label{thm:main}
Let $D=(V,A)$ be a digraph of order~$n$. 
Then {\em Min-Greedy} always finds an acyclic set in $D$ 
of size at least $\sum_{v\in V} \left(d^+(v)+1\right)^{-1}$.
\end{thm}
\begin{proof}
Using the notation from the algorithm, 
let $v_i$ be the selected vertex of 
minimum outdegree in $D_i$. Then for all vertices 
$w\in N^+_{D_i}(v_i) \cup \{v_i\}$ holds 
\begin{align*}
d^+_{D}(w)+1 \ge d^+_{D_i}(w)+1 \ge d^+_{D_i}(v_i)+1 = \left|N^+_{D_i}(v_i) \cup \{v_i\}\right|.
\end{align*}  
Thus, 
\begin{align*}
\sum_{w\in N^+_{D_i}(v_i) \cup \{v_i\}}\left(d^+_{D}(w)+1\right)^{-1} 
& \le \sum_{w\in N^+_{D_i}(v_i) \cup \{v_i\}}\left(d^+_{D_i}(w)+1\right)^{-1} \\
& \le \sum_{w\in N^+_{D_i}(v_i) \cup \{v_i\}}\left(\left|N^+_{D_i}(v_i) \cup \{v_i\}\right|\right)^{-1} = 1. 
\end{align*}
On the other hand, since the algorithm partitions the vertex set 
of~$D$ into a disjoint union of subsets as 
\[V(D) = \bigcup_{i=1}^{|S|} \left(N^+_{D_i}(v_i) \cup \{v_i\}\right),\] 
we have 
\begin{align*}
\sum_{v\in V} \left(d^+_D(v)+1\right)^{-1} & = \sum_{i=1}^{|S|}\sum_{w\in N^+_{D_i}(v_i) \cup \{v_i\}}\left(d^+_{D}(w)+1\right)^{-1}
\le \sum_{i=1}^{|S|} 1 = |S|, 
\end{align*}
as desired.
\end{proof}

Just like the Caro--Wei bound~\cite{Caro79,Wei81} for the independence number 
of undirected graphs implies the Tur\'{a}n bound~\cite{Turan41} by the 
inequality of arithmetic and harmonic means, we 
have the following simple bound on the size of a maximum 
acyclic set, and hence, on the directed feedback vertex number, 
in terms of average outdegree: 

\begin{cor}\label{cor:ollary}
Let $D=(V,A)$ be a digraph of order~$n$ and of average 
outdegree $\overline{d^+}$. Then
\[ \tau_0(D) \le n\cdot\left(\,1 - \left(\,\overline{d^+}+1\right)^{-1}\,\right). \]
\end{cor}
\begin{proof}
We show the equivalent statement that the digraph~$D$ 
has an acyclic set of cardinality
at least $n / \left(\,\overline{d^+}+1\right)$. The bound 
$\sum_{v\in V} \left(d^+(v)+1\right)^{-1}$ from Theorem~\ref{thm:main} 
looks different at first glance. Nevertheless, it easily implies a
bound in terms of average outdegree: 
recall that the {\em inequality of the harmonic, geometric and 
arithmetic mean} (see~\cite[Chapter~16]{AZ01}) 
states that the geometric mean of~$n$ 
positive numbers $a_1,a_2,\ldots,a_n$ is sandwiched between 
the harmonic mean and the arithmetic mean of these numbers, 
that is, 
\[\frac{n}{\sum_{i=1}^n a_i^{-1}} \le \left(\prod_{i=1}^n a_i \right)^{1/n} \le \sum_{i=1}^n a_i/n.  \]
Now choose the $a_i$ to be the vertex degrees in~$D$ 
increased by~$1$ each. Then the outermost inequality yields
$\frac{n}{\sum_{v\in V} \left(d^+(v)+1\right)^{-1}} \le \sum_{v\in V}(d^+(v)+1)/n$. 
A very simple calculation completes the proof. 
\end{proof}
Both the bound from Theorem~\ref{thm:main}
and the one from Corollary~\ref{cor:ollary} 
are sharp, as witnessed, for example, 
by the digraph of order~$k \cdot m$ that is obtained as the 
disjoint union of~$m$ many $k$-cliques.

Notice that we obtain the Caro--Wei bound and the Tur\'{a}n bound, 
respectively, if we restrict the scope of the 
above statements to symmetric digraphs: for these, the 
size of a maximum acyclic set is equal to the independence 
number, and the outdegree of each vertex is equal to its degree 
(which in turn is equal to half its total degree).

As a final note, we remark that the Moon--Moser bound 
on the number of maximal independent sets in 
undirected graphs~\cite{MM65} 
does not generalize to an analogous statement 
about maximal acyclic sets; as a matter of fact, 
not even tournaments allow for a clear 
generalization~\cite{Moon71}. In the undirected 
case, the proofs of both the Caro--Wei bound and 
the Moon--Moser bound can be used to derive 
Tur\'{a}n's graph theorem, see~\cite{Aigner95}. 
A general theme for further research is to identify
those fragments of the theory of 
undirected graphs that generalize smoothly to the 
case of digraphs. 


\end{document}